\documentclass[journal]{IEEEtran}
\IEEEoverridecommandlockouts
\usepackage{amsmath,graphicx}
\usepackage{amsfonts}
\usepackage{amsbsy}
\usepackage{amsmath,amssymb}
\usepackage{amsthm}
\usepackage{times}
\usepackage{graphicx}
\usepackage[caption=false,font=normalsize,labelfont=sf,textfont=sf]{subfig}

\usepackage{enumerate}
\usepackage[usenames]{color}
\usepackage{epstopdf}
\usepackage{cite}
\usepackage{epsfig}
\usepackage{psfrag}
\usepackage{xcolor}
\usepackage{bm}
\usepackage{epstopdf}
\usepackage{cite}
\usepackage{color}
\usepackage{cuted}
\usepackage{xcolor}
\usepackage{verbatim}
\usepackage{algorithm}
\usepackage{algorithmic}
\usepackage{booktabs}
\usepackage{colortbl}
\usepackage{subfig}
\usepackage{enumitem}
\usepackage[colorlinks =True,
            linkcolor  =blue,
            anchorcolor=green,
            citecolor  =blue]{hyperref}

\newtheorem{theorem}{Theorem}[section]
\newtheorem{definition}{Definition}

\newtheorem{lemma}[theorem]{Lemma}
\newtheorem{remark}[theorem]{Remark}
\def\BibTeX{{\rm B\kern-.05em{\sc i\kern-.025em b}\kern-.08em
		T\kern-.1667em\lower.7ex\hbox{E}\kern-.125emX}}
\begin{document}
	
	\title{Neural Collapse based Deep Supervised Federated Learning for Signal Detection in OFDM Systems
	}
	\author{Kaidi Xu, Shenglong Zhou, and Geoffrey Ye Li, {\it IEEE Fellow}
		
    \thanks{Kaidi Xu and Geoffrey Ye Li are with the ITP Lab, Department of EEE, Imperial College London, UK. Shenglong Zhou is with the School of Mathematics and Statistics, Beijing Jiaotong University, China. Emails:  k.xu21@imperial.ac.uk, shlzhou@bjtu.edu.cn, geoffrey.li@imperial.ac.uk}
			\thanks{
				*Corresponding author: Shenglong Zhou. This work was supported by the Fundamental Research Funds for the Central Universities.
			}}

		\maketitle
		
		\begin{abstract}
        Future wireless networks are expected to be AI-empowered, making their performance highly dependent on the quality of training datasets. However, physical-layer entities often observe only partial wireless environments characterized by different power delay profiles. Federated learning is capable of addressing this limited observability, but often struggles with data heterogeneity.  To tackle this challenge, we propose a neural collapse (NC) inspired deep supervised federated learning (NCDSFL) algorithm.  
Specifically, we first define an NC solution for the multi-binary classification problem and establish its optimality by revealing its connection to the global optimum. We then incorporate the deep supervision technique into deep neural networks and fix the weights at both the output layer and an auxiliary hidden layer using the derived NC solutions.  When this strategy is applied in a federated learning setting, it encourages different clients to produce similar hidden features, thereby enabling the proposed algorithm to effectively address data heterogeneity and achieve fast convergence. Simulations for signal detection in OFDM systems confirm the NC phenomenon and demonstrate that NCDSFL outperforms several baselines in terms of convergence speed and accuracy.
        

		\end{abstract}
		
		
		\section{Introduction}
        \IEEEPARstart{D}{eep}  Learning (DL) has been recognized as a powerful tool to enhance wireless communication systems in various applications, such as resource allocation, networking, and mobility management, and signal detection  \cite{sun2019application, pham2021intelligent, ye2017power, yi2020deep, ye2020deep,ye2021deep, honkala2021deeprx, sheng2024beam}.        
        Unlike conventional model-based methods, DL-based methods can learn implicit information from the training data and jointly optimize different communication modules wherever gradient backpropagation is possible.
        In \cite{ye2017power}, a deep neural network (DNN) is used as the neural receiver for joint channel estimation and signal detection in an Orthogonal Frequency Division Multiplexing (OFDM) system.        
        In \cite{yi2020deep}, a channel estimation DNN replaces the interpolation procedure to exploit implicit subcarrier correlations.
        End-to-end designs are later introduced in \cite{ye2020deep,ye2021deep}, where DNNs serve as both transmitter and receiver.
        However, the performance of DL-based methods heavily depends on the quality of the training dataset. As illustrated in Fig. \ref{sys_diagram}, user devices (UDs) can only observe partial wireless environments, characterized by unique power delay profiles (PDPs) determined by the positions of UDs, the base station, and the scatterers. 
        Collecting data from all users is expensive and raises privacy concerns, limiting the practicality of DL-based methods.

     \begin{figure}[!t]
			\centering
			\includegraphics[width=0.45\textwidth]{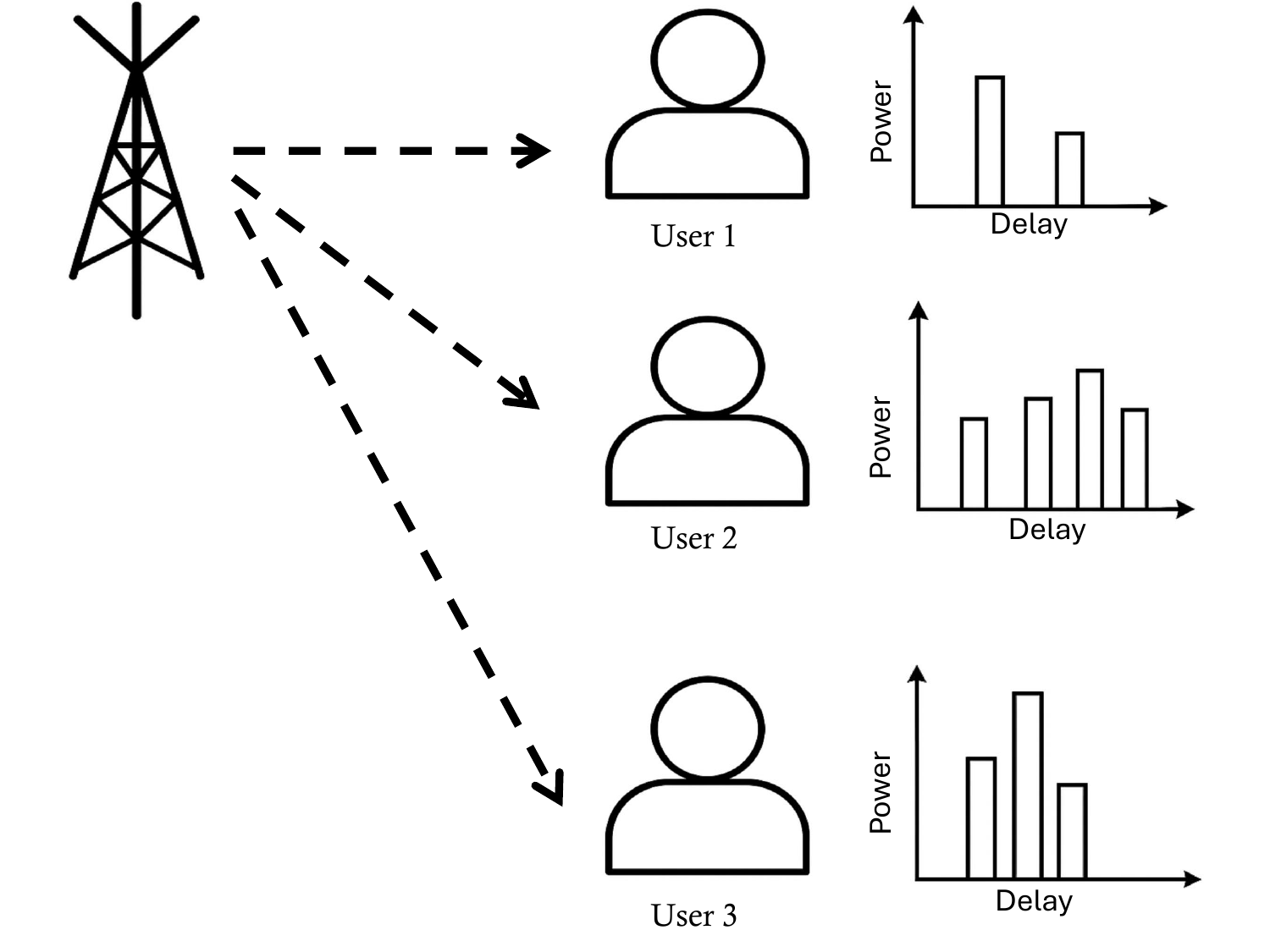}
			\caption{System diagram}
			\label{sys_diagram}
		\end{figure}
		
        To address this issue, federated learning (FL) has been extensively applied to wireless communications \cite{mcmahan2017communication,qin2021federated}.
        In a typical FL system, there is a central server aggregating and broadcasting model parameters from and to the clients without raw data exchange.
        The aggregated model integrates knowledge from the clients with partial observations of the wireless environment, enabling good generalizability across the global environment.
        For instance, in \cite{xu2024federated, xu2024rescale}, FL aggregates knowledge from client agent models to train a global agent model in Vehicle-to-Everything (V2X) systems.
        An online personalized FL algorithm has also been designed to collaboratively train a neural receiver across multiple cells \cite{wang2023new}.
        However, clients' local datasets are typically heterogeneous due to their limited observability, which can slow down the convergence of FL algorithms and increase communication overhead \cite{li2020federated}.
        
        To mitigate client drift caused by data discrepancy, feature alignment has been explored in various studies \cite{gupta2025fedalign, li2020federated, li2021model, wang2020federated}.
        FedProx \cite{li2020federated} introduces a regularization term that penalizes the divergence between the local client model and the global model during local updates, implicitly aligning their feature representations.
        MOON \cite{li2021model} employes contrastive learning to promote similarity between the feature representations of the local and global models.
        In \cite{wang2020federated}, hidden elements are permuted within each layer to align elements with similar feature extraction patterns for more effective aggregation.
        While these methods demonstrate the effectiveness of feature alignment in FL, they often incur additional computational overhead or treat all model parameters uniformly, overlooking the influence of data distribution.
        To deploy FL effectively in practical wireless systems, it is crucial to reduce both communication and computation overhead.

		To achieve this goal, we propose a neural collapse (NC) inspired deep supervised FL (NCDSFL) framework. The main contributions are threefold.
    \begin{itemize}[leftmargin=13pt]           
\item[1)]  We provide an explicit NC solution for a multi-binary classification problem and rigorously establish its optimality by revealing the relationship between an NC  and a global solution to the problem.  To the best of our knowledge, this has not been addressed in the existing literature.  

\item[2)] We then introduce the technique of deep supervision (DS) in a DNN by adding one auxiliary layer connected to the third-to-last layer. This approach can be interpreted as incorporating a regularization term into the objective function of the target optimization model. During training, we fix the weights at both the output layer and the auxiliary layer using two derived NC solutions. As a result, the number of trainable parameters is reduced. 

\item[3)] We integrate DS-based DNNs, with a portion of their weights fixed using NC solutions, into an FL paradigm. This design encourages different clients to produce similar hidden features at both the penultimate and auxiliary layers. A hidden feature refers to the output of a specific layer of neurons; see \eqref{def-hidden-feature} for a detailed definition.
Consequently, the proposed NCDSFL algorithm can effectively address data heterogeneity and achieve fast convergence, as demonstrated by simulation results on signal detection in OFDM systems compared to several baseline methods. Furthermore, we empirically validate the emergence of the NC phenomenon in multi-binary classification tasks. 
\end{itemize}

        \subsection{Related Works}
        NC \cite{papyan2020prevalence} is a terminal state of neural networks at the end of the training phase.
        It provides an elegant mathematical analysis on the features and parameters of the last layer by peeling them off the neural networks and treating them as free variables. 
        A specific NC form depends on the formulation of the training problem, but usually includes several basic properties. More details can be found in  \cite{kothapalli2022neural}.

  A growing body of research has explored the idea of NC under various configurations. For example, in \cite{papyan2020prevalence}, NC has been first defined in the context of canonical classification with cross-entropy loss. 
  In \cite{zhu2021geometric}, the loss landscape with weight decay has been further analyzed. 
  It demonstrated that only the global minimizers correspond to simplex Equiangular Tight Frames (ETFs) while all other critical points are strict saddle points. This insight explains the effectiveness of cross-entropy loss for classification tasks.
In \cite{nguyen2022memorization}, the effect of label corruption was examined, and it has been shown that NC can still emerge, albeit with modified geometric structures due to memorization.
The authors in \cite{yang2022inducing} analyzed the gradients of the cross-entropy loss and proposed replacing the final learnable classifier with fixed ETF classifiers and a custom loss function, which can achieve a theoretical convergence rate no worse than that of the original configuration. In \cite{li2023neural}, NC in multi-label classification was investigated using a pick-all-label loss, where each sample can belong to multiple classes. Beyond multi-label classification, NC has recently been adopted for multivariate regression \cite{andriopoulos2024prevalence}.

Deep learning-based signal detection can be modeled as a multi-binary classification problem, which is equivalent to multi-label classification.
However, the method in \cite{li2023neural} required a threshold to determine how many labels to assign in the prediction phase, making it unsuitable for our signal detection setting.
To this end, we derive a new NC form for multi-binary classification and prove its global optimality.

DS was first introduced in \cite{lee2015deeply}, where additional supervision signals were applied to intermediate hidden layers of DNNs. This approach helps mitigate the vanishing gradient problem, encourages the learning of more discriminative features, and improves overall network performance \cite{li2022comprehensive}.
It has been applied to a variety of computer vision tasks, including image segmentation, object detection, and image super-resolution \cite{dou20163d, wang2021sne, kim2016deeply}.
The MOON algorithm \cite{li2021model} can also be considered a form of DS, since it introduces regularization terms for intermediate hidden representations, jointly optimized with the main loss function.  Since NC characterizes the optimal solution to the training problem, it is promising to train DNNs using fixed NC solutions as weights not only at the penultimate layer but also at the auxiliary layer. However, to the best of our knowledge, existing DS methods have not yet incorporated NC-based solutions.

\subsection{Organizations}
The paper is organized as follows.   Section \ref{section_mbclpm} introduces the multi-binary classification problem, constructs a layer-peeled model, defines the corresponding NC solution, and establishes its optimality.
Section \ref{ncdsfl_framework_section} presents the FL framework and develops the proposed NCDSFL algorithm, with an application to signal detection in OFDM transmission systems.
Simulation results are provided in Section \ref{section_simulation},  and conclusions are drawn in Section \ref{section_conclusion}.

\section{Multi-Binary Classification}\label{section_mbclpm}              
In this section, we introduce the multi-binary classification problem and propose the NC solution of the corresponding layer peeled model, followed by its application in DNNs.

\subsection{Problematic Description}       
Let ${s\in\mathcal{S}^I}$ denote a binary label sequence of $I$ binary classes, where $\mathcal{S}^I$ denotes the space of binary sequences with length $I$. The $i$-th value, denoted by $s_i$, of $s$ represents the $i$th class of the data sample labeled by $s$. Taken ${I=3}$ as an example, ${\mathcal{S}^I=\{000,001,010,011,100,101,110,111\}}$ and ${s=011}$ means that the sample labeled as $s$ belongs to the second or third classes due to ${s_2=s_3=1}$. In the sequel, for notational simplicity, we denote  \begin{equation}
    \mathcal{S}=\mathcal{S}^I=\left\{s^{(0)},s^{(1)},\ldots,s^{(2^I-1)}\right\},
\end{equation} as $\mathcal{S}^I$ has $2^I$ elements. The superscript $j$ in $s^{(j)}$ is determined by binary sequence $s$ itself. For example, $j=1$ if $s=001$ and $j=6$ if $s=110$.  
We consider the balanced distribution case, where there are $K$ samples for each label $s\in\mathcal{S}^I$\footnote{This assumption is reasonable in a wireless communication scenario, because the transmitted binary sequences are usually uniformly distributed binary sequences.}.
As a result, there are $K2^I$ data samples in total.
In addition, we denote the linear classifiers for the $i$th binary class as $\mathbf{w}_{i,1}\in\mathbb{R}^d$ and $\mathbf{w}_{i,0}\in\mathbb{R}^d$ and the hidden feature for the $k$th data sample labeled by sequence $s$ as $\mathbf{h}_s^{(k)}\in\mathbb{R}^d$.


Following the canonical classification paradigm, we use cross entropy with regularized weights and features as our loss function. 
 The resulting loss function is given by
    \begin{equation}\label{loss}
    \begin{aligned}
        L(\mathbf{W},\mathbf{H})&=\lambda \|\mathbf{W}\|^2+\lambda \|\mathbf{H}\|^2\\
        &+\frac{1}{KI2^I}\sum_{k=1}^K\sum_{i=1}^I \sum_{s\in\mathcal{S}}\ln\Big(1+\exp\Big(\phi_{kis}(\mathbf{W},\mathbf{H})\Big)\Big),
    \end{aligned}
    \end{equation} 
    where $\lambda>0$ and $\phi_{kis}(\mathbf{W},\mathbf{H})$ is defined by
\begin{equation}\label{loss-phi}
    \begin{aligned}
        \phi_{kis}(\mathbf{W},\mathbf{H})= \left\langle\mathbf{w}_{i,1-s_i}-\mathbf{w}_{i,s_i},\mathbf{h}_s^{(k)}\right\rangle.
    \end{aligned}
    \end{equation}            
Here, $\langle\cdot,\cdot\rangle$ denotes the vector inner product, $1-{s}_i, s_i\in\{0,1\}$, and
$\mathbf{W}$ and $\mathbf{H}$ denote the collection of all classifiers and all features, namely,
\begin{equation}\label{def-W-H}
\begin{aligned}
\mathbf{W}&=\Big[\mathbf{w}_{1,0}, \mathbf{w}_{2,0},\ldots,\mathbf{w}_{I,0}, \mathbf{w}_{1,1}, \mathbf{w}_{2,1},\ldots\mathbf{w}_{I,1}\Big],\\
\mathbf{H}&=\Big[ 
\mathbf{h}_{s^{(0)}}^{(1)}, \ldots,  \mathbf{h}_{s^{(2^I-1)}}^{(1)},  \ldots, \mathbf{h}_{s^{(0)}}^{(K)},\ldots, \mathbf{h}_{s^{(2^I-1)}}^{(K)}\Big].            
\end{aligned}
\end{equation} 
\subsection{NC Solutions} 
It is known that a layer-peeled model  \cite{fang2021exploring,yang2022inducing} ignores the impact of the backbone network and treats the features of the penultimate hidden layer $\mathbf{H}$ as free variables, which is true when the DNN is over-parameterized due to the universal approximation property of DNNs.
The layer peeled model of the multi-binary classification problem can be expressed as the following optimization problem,   
\begin{equation}\label{multi-binary_classification}
    \min_{\mathbf{W},\mathbf{H}}\quad L(\mathbf{W},\mathbf{H}).
\end{equation}
To analyze the above problem, let ${\rm vec}(\mathbf{H})$ denote the column-wise vectorization of $\mathbf{H}$, namely, 
\begin{equation}
    {\rm vec}(\mathbf{H})=\Bigg[\mathbf{h}_{s^{(0)}}^{(1)T},\ldots, \mathbf{h}_{s^{(2^I-1)}}^{(1)T}, \ldots, \mathbf{h}_{s^{(0)}}^{(K)T},\ldots, \mathbf{h}_{s^{(2^I-1)}}^{(K)T}\Bigg]^T,          
\end{equation}
where $\mathbf{h}_{s^{(i)}}^{(k)T}=(\mathbf{h}_{s^{(i)}}^{(k)})^T$.         
Let $\mathbf{I}_d\in\mathbb{R}^{d\times d}$ be an identity matrix, and $\mathbf{A}\in\mathbb{R}^{dI\times dK2^I}$ and $\mathbf{C}^i\in\mathbb{R}^{d\times d2^i}$ be defined by,
\begin{equation}\label{A_def}
    \mathbf{A} = \underbrace{\left[\mathbf{B}^I,\ldots,\mathbf{B}^I\right]}_{\text{K times}}, \qquad \mathbf{C}^i=\underbrace{[\mathbf{I}_d,\ldots,\mathbf{I}_d]}_{2^{i}~ \text{times}},
\end{equation}
where $\mathbf{B}^I$ is recurrently generated as follows,
\begin{equation}\label{B_def}
            \mathbf{B}^1=\begin{bmatrix}
                -\mathbf{I}_d & \mathbf{I}_d 
            \end{bmatrix},
            \quad
            \mathbf{B}^{i+1}=\begin{bmatrix}
                \mathbf{B}^i &\mathbf{B}^i\\
                -\mathbf{C}^i &\mathbf{C}^i
            \end{bmatrix},
        \end{equation}  
        for ${i=1,2,\cdots,I-1}$. Based on these notations, an NC solution to problem \eqref{multi-binary_classification} is introduced as follows. 
        \begin{definition}\label{nc_def} A point $(\mathbf{W},\mathbf{H})$ is called an NC solution to problem (\ref{multi-binary_classification}) for the multi-binary classification if it satisfies the following conditions:
            \begin{itemize}[leftmargin=12pt]
            \item {\it (NC1: hidden feature collapse).}
            All features  converge to the mean, i.e., $\mathbf{h}_s^{(k)} = \frac{1}{K}\sum_{j=1}^K \mathbf{h}_s^{(j)}, \forall s, k$.
            \item {\it (NC2, classifier collapse).}
            Classifiers of all parallel binary classification problems form an orthogonal set, i.e., ${\langle\mathbf{w}_{i,1},\mathbf{w}_{j,1}\rangle=0},\forall i\neq j$, and ${\mathbf{w}_{i,1}+\mathbf{w}_{i,0}=\mathbf{0}}, \forall i$.
            In addition, ${\|\mathbf{w}_{i,1}\|=\|\mathbf{w}_{j,1}\|},\forall i,j$.
            \item {\it (NC3: duality).} 
              The classifier and the hidden features are linearly aligned by 
              \begin{equation}
                  {{\rm vec}(\mathbf{H})=c_0\left[\mathbf{A}^T,-\mathbf{A}^T\right]{\rm vec}(\mathbf{W})},
              \end{equation} for a scaling constant $c_0>0$.
        \end{itemize}
        \end{definition}        
        The following theorem reveals the relationship between the NC  solution and the global solution to problem \eqref{multi-binary_classification}.
        \begin{theorem}\label{NC_theorem}
            If ${0<\lambda <  {1}/({2I\sqrt{2K2^I}}})$, then any global minimizer of problem (\ref{multi-binary_classification}) is an NC solution.
        \end{theorem}
        \begin{proof}
            The proof of Theorem \ref{NC_theorem} is given in Appendix \ref{proof}.
        \end{proof}
        \begin{remark} The proposed NC for multi-binary classification problem can also be extended to a multi-label classification scenario, where each data sample has an arbitrary number of class labels from $I$ classes.
        This multi-label classification problem is equivalent to the multi-binary classification problem with $I$ parallel binary classification sub-tasks, where each binary classification represents whether the sample is classified into the corresponding class or not.
        \end{remark} 
        
  \subsection{Deep supervised DNNs with NC Weights}

    Now we introduce $O$-layer neural networks with $(O-1)$ hidden layers. Specifically, let $d_o$ be the number of hidden units of the $o$-th hidden layer for ${i\in[O-1]}$, where $[O]=\{1,2,\cdots,O\}$. Let $d_0$ and $d_O$ represent the number of input and output units. Denote ${\mathcal{W}=\{\textbf{W}_1,\textbf{W}_2,\cdots,\textbf{W}_{O}\}}$ with $\textbf{W}_o\in \mathbb{R}^{d_{o}\times d_{o-1}}$ being the weight matrix of the $o$th layer. 
    For simplicity, we integrate both the weights and the bias into $\mathbf{W}$.    
    Let ${\mathcal{D}=\{(\mathbf{x}_n,\mathbf{y}_n): n\in[N]\}}$ be a given dataset, where $\mathbf{x}_n\in\mathbb{R}^{d_0}$ and $\mathbf{y}_n\in\mathbb{R}^{d_{O}}$ is the $n$th sample's feature and label, and $N$ is the number of samples. 
    The optimization model of DNNs can be built as follows,
\begin{align}
\label{OP31} 
&\min\limits_{\mathcal{W}} ~ L( \textbf{Y}, \widetilde{\textbf{Y}})  \\
&{~\rm{s.t.}}~~~ \widetilde{\textbf{Y}}={\rm softmax}(\textbf{W}_{O}\sigma(\textbf{W}_{O-1}\cdots\sigma(\textbf{W}_2\sigma(\textbf{W}_1\textbf{X})))),\nonumber
\end{align}
where $L$ is a loss function, $\sigma(\cdot)$ is an activation function, such as the rectified linear unit (ReLU) used in the sequel, $\textbf{X}=[\textbf{x}_1,\cdots, \textbf{x}_N]$, and $\textbf{Y}=[\textbf{y}_1,\cdots, \textbf{y}_N]$.   For simplicity, let
\begin{align}
\label{def-hidden-feature} \textbf{H}_{0}= \textbf{X},~~\textbf{H}_{o}=\sigma(\textbf{W}_o\textbf{H}_{o-1}),~ o=1,2,\ldots,O-1.\end{align}
Then $\textbf{H}_{o}$ is the output of the neurons at the $o$th layer and is referred to as the hidden features at this layer.
 
               \begin{figure}[!t]
			\centering
			\includegraphics[width=0.49\textwidth]{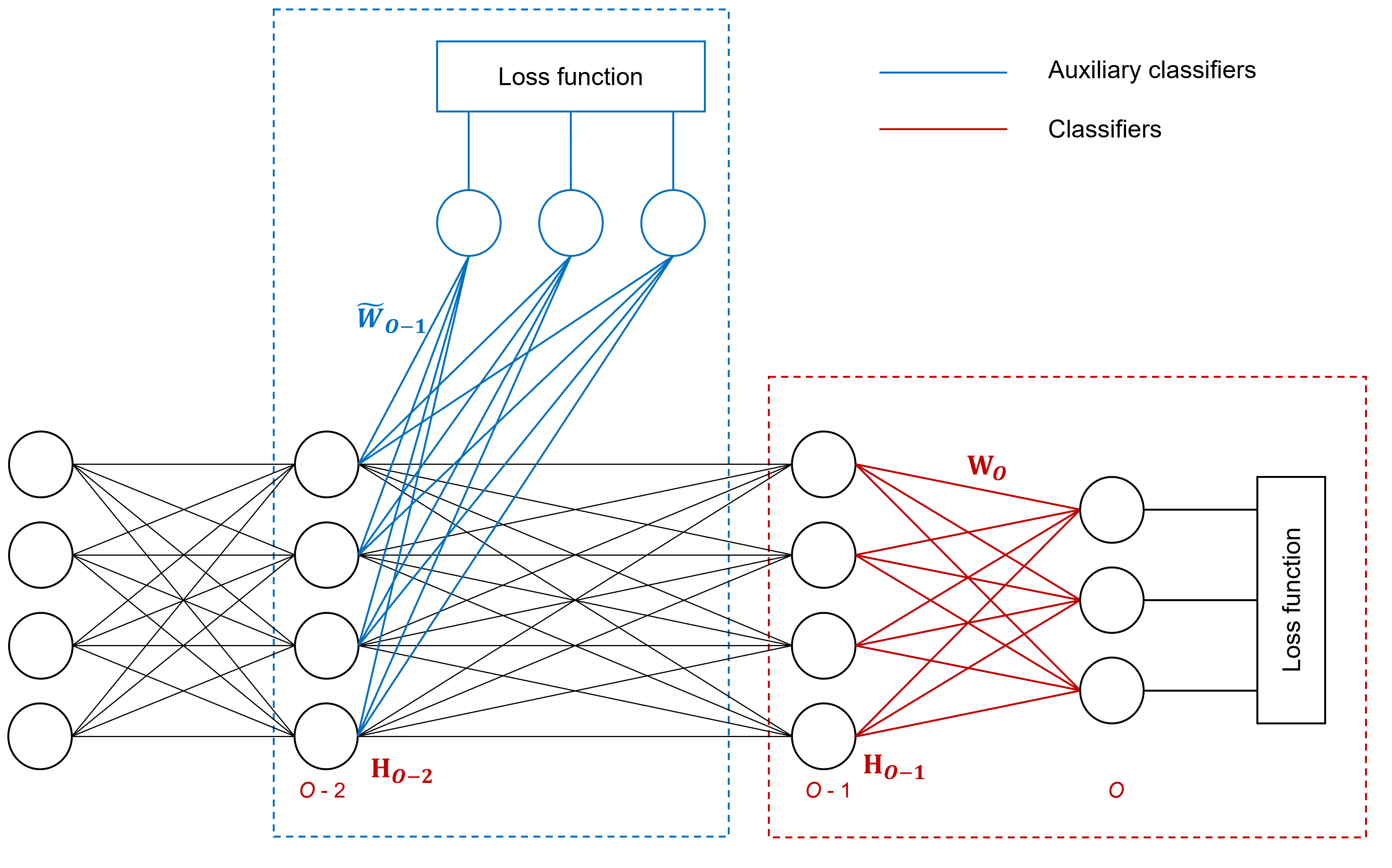}
			\caption{A deep supervised DNN}
			\label{system}
		\end{figure}

        The motivation for considering the NC solution stems from observations in many applications, such as the signal detection task illustrated in Figs. \ref{nc2} and \ref{nc3}, where the penultimate hidden layer of a DNN exhibits classifier collapse (i.e., NC2), as well as classifier and feature linear alignment (i.e., NC3), and the NC-related studies introduced in the introduction section.
        To leverage this structural behavior, we embed NC solutions as the weights at the penultimate hidden layer of a DNN. 
        
\subsubsection{Direct embedding}
 The first embedding is employed directly on the original DNN at the $O$th layer, i.e., the output layer, see the red rectangle in  Fig. \ref{system}, where $\textbf{W}_{O}$ in model \eqref{OP31} will be fixed by using an NC weight defined in Definition \ref{nc_def}.
 Specifically, we generate a set of orthogonal vectors $\{\mathbf{w}_{i,0}\}$ and let ${\mathbf{w}_{i,1}=-\mathbf{w}_{i,0}}$ for all ${i\in[I]}$ to form an NC weight $\textbf{W}_{NC}$. 
 The weight is then fixed during the whole training phase, as fixed NC classifiers do not harm classification performance \cite{zhu2021geometric}.
 To do so, the constraint in model \eqref{OP31}  becomes
\begin{align}
\label{OP31-con} 
\widetilde{\textbf{Y}}={\rm softmax}(\textbf{W}_{NC}\sigma(\textbf{W}_{O-1}\cdots\sigma(\textbf{W}_2\sigma(\textbf{W}_1\textbf{X})))).
\end{align}     
\subsubsection{Auxiliary embedding under the DS framework}  The second embedding is applied for an auxiliary layer. We introduce one auxiliary layer with weight matrix $\widetilde{\textbf{W}}_{O-1}$ collecting the $(O-2)$th neuron in the original network. This idea is termed as deep supervision (DS) in \cite{li2022comprehensive}.  
As shown in the blue rectangle in  Fig. \ref{system}, the $(O-2)$th neuron output in the original network is connected to the auxiliary loss function by the weight matrix $\widetilde{\textbf{W}}_{O-1}$. Again,  weight $\widetilde{\textbf{W}}_{O-1}$ is fixed by another NC weight  $\widetilde{\textbf{W}}_{NC}$ defined similarly as above. This indicates that we can add a regularization term  in model \eqref{OP31}, 
     \begin{equation} 
 {\rm CE} \big( \textbf{Y},  {\rm softmax} (\textbf{H}_{O-2} \widetilde{\textbf{W}}_{NC})\big),
 \end{equation}
 where ${\rm CE}(\cdot)$ is a cross entropy loss. This together with \eqref{OP31-con}, model \eqref{OP31} turns to have the following form,
\begin{align}
\label{OP31-reg} 
&\min\limits_{\widetilde{\mathcal{W}}} ~L( \textbf{Y}, \widetilde{\textbf{Y}})  + \mu {\rm CE}( \textbf{Y},  \textbf{H}_{O-2} \widetilde{\textbf{W}}_{NC})\\
&{~\rm{s.t.}}~~~ \widetilde{\textbf{Y}}={\rm softmax}(\textbf{W}_{NC}\sigma(\textbf{W}_{O-1}\cdots\sigma(\textbf{W}_2\sigma(\textbf{W}_1\textbf{X})))),\nonumber
\end{align}
where ${\widetilde{\mathcal{W}}=\{\textbf{W}_1,\textbf{W}_2,\cdots,\textbf{W}_{O-1}\}}$ and ${\mu>0}$. In the next section, we will explore the above model to develop an efficient algorithm for signal detection problems.

        \section{Neural Collapse inspired Deep Supervised Federated learning} \label{ncdsfl_framework_section}
        
        In this section, we first introduce the framework of FL and develop the  NCDSFL algorithm based on it.

        \subsection{Federated Learning}
        A FL system consisting of one central server and $J$ clients aims to learn a shared model for all clients to minimize the total loss function without raw data exchange.
        Specifically, this learning process can be formulated as the following optimization problem,
        \begin{equation}\label{FL_opt}
            \begin{aligned}
                \min_{\mathbf{v},\mathbf{v}^i\in\mathcal{F},i\in[J]} \quad & \sum_{i\in[J]} \Big(L(\mathbf{v}^i;\mathcal{D}^i)+R(\mathbf{v}^i;\mathcal{D}^i)\Big),\\
                {\rm s.t.}~~~~\quad&~ \mathbf{v}^i=\mathbf{v},~~i\in[J],
            \end{aligned}
        \end{equation}
        where $\mathbf{v}^i$, $\mathbf{v}$, $\mathcal{F}$ and $\mathcal{D}^i$ denote the trainable local parameters at client $i$, the shared global parameters, the feasible region, and the dataset of client $i$, respectively,
        $L(\cdot;\cdot)$ denotes the local loss function used to train the model locally and depends on the local parameters and the local dataset, and $R(\cdot;\cdot)$ is a regularization term added to a local loss to guide training.

In a canonical centralized FL system, the learning process repeats the following three key steps until the model parameters converge or some terminating conditions are met:
        \begin{itemize}[leftmargin=12pt]
        \item \textbf{Local update}: Each client updates its local model based on its local parameters and dataset by an optimization algorithm such as stochastic gradient descent or Adam \cite{kingma2014adam}.
        \item \textbf{Aggregation}: After several local updates, all clients or a part of them send their trained local model parameters to the server. The server aggregates the models by averaging the received parameters.
        \item \textbf{Broadcast}: The server then broadcasts the aggregated global model to all clients.
           \end{itemize} 
  \subsection{DS DNNs with NC Weights in FL}       
         Note that in FL, model aggregation occurs after several local updates. This can lead to client drift and slow down the convergence, especially when local datasets are heterogeneous.
To facilitate more efficient aggregation, a natural approach is to align hidden features across clients.
According to the NC theory \cite{kothapalli2022neural}, the hidden features (e.g., $\textbf{H}_{O-1}$ in the direct embedding and $\textbf{H}_{O-2}$ in the auxiliary embedding)  and their classifiers (i.e., $\textbf{W}_{NC}$ and $\widetilde{\textbf{W}}_{NC}$) eventually collapse (i.e., satisfying NC1 and NC2) and become linearly aligned (i.e., satisfying NC3).  At the same time, DS \cite{li2022comprehensive} provides direct guidance to intermediate hidden features (e.g., $\textbf{H}_{O-2}$).

By integrating these two insights, we can guide and align hidden feature distributions across clients, thereby promoting consistency among local models. To achieve this, we incorporate model \eqref{OP31-reg} into the FL framework. Specifically, for $J$ clients, we equip each client $i$ with $E$ DNNs with ${E\geq1}$, which have different learnable parameters. Therefore, there are $EJ$ DNNs in total. The dataset, $\mathcal{D}^i$, for client $i$ is
 \begin{align}
        &\mathcal{D}^i= \left\{\mathcal{D}^{i,1},\mathcal{D}^{i,2},\cdots,\mathcal{D}^{i,E}\right\},\\
         &\mathcal{D}^{i,e}= \left\{(\mathbf{x}_1^{i,e},\mathbf{y}_1^{i,e}),\cdots,(\mathbf{x}_N^{i,e},\mathbf{y}_N^{i,e})\right\},~~e\in[E],
        \end{align}
where $\mathcal{D}^{i,e}$ is the dataset for the $e$th DNN for client $i$. Let \begin{equation}
    \textbf{X}^{i,e}=\left[\textbf{x}_1^{i,e},\cdots, \textbf{x}_N^{i,e}\right], ~~\textbf{Y}^{i,e}=\left[\textbf{y}_1^{i,e},\cdots, \textbf{y}_N^{i,e}\right]
\end{equation} be the input and the label matrix for the $e$th DNN of client $i$.
Then variables  in \eqref{FL_opt} for each ${i\in[J]}$ are specified by,
        \begin{align}
        &\mathbf{v}^i=\left\{\widetilde{\mathcal{W}}^{i,1},\widetilde{\mathcal{W}}^{i,2},\cdots,\widetilde{\mathcal{W}}^{i,E}\right\},\\
       & \widetilde{\mathcal{W}}^{i,e} =\left\{\textbf{W}_1^{i,e},\textbf{W}^{i,e}_2,\cdots,\textbf{W}^{i,e}_{O-1}\right\},~~e\in[E],
        \end{align}
        where $\widetilde{\mathcal{W}}^{i,e} $ collects all weights to be trained for the $e$th DNN of client $i$. Then functions in \eqref{FL_opt} for each $i$ are specified by,     
                \begin{align}
  & L(\mathbf{v}^i;\mathcal{D}^i) = \sum_{e\in[E]}  L\left( \textbf{Y}^{i,e}, \widetilde{\textbf{Y}}^{i,e}\right),\\
     &  R(\mathbf{v}^i;\mathcal{D}^i) = \mu \sum_{e\in[E]}   {\rm CE}\left( \textbf{Y}^{i,e},  \textbf{H}_{O-2}^{i,e} \widetilde{\textbf{W}}_{NC}\right),
        \end{align}
        where  $ \widetilde{\textbf{Y}}^{i,e}$ is given similar to that in \eqref{OP31-reg}, namely,              
        \begin{align*}
  \widetilde{\textbf{Y}}^{i,e}={\rm softmax}\left(\textbf{W}_{NC}\sigma(\textbf{W}_{O-1}^{i,e}\cdots\sigma(\textbf{W}_2^{i,e}\sigma(\textbf{W}_1^{i,e}\textbf{X}^{i,e})))\right).
        \end{align*}
One can observe that for each client $i$, we use two fixed NC weights: $\textbf{W}_{NC}$ fixed at the penultimate hidden layer for all $E$ primal networks and $\widetilde{\textbf{W}}_{NC}$ fixed at the auxiliary layer for all $E$ auxiliary networks. These two weights are irrelevant to client $i$ and can be calculated in advance according to NC2.

        \subsection{The NCDSFL Algorithm}

        \begin{algorithm}[!t]
			\caption{NCDSFL: Neural Collapse inspired Deep Supervised Federated Learning} 
			\begin{algorithmic}[1] \label{NCDSFL}
				\STATE	 \textbf{Initiation}: Initialize learnable parameters $\{\mathbf{v}^{i,0},i\in[J]\}$ for all clients and ${\mathbf{v}^{0} = ({1}/{[J]})\sum_{i=1}^J \mathbf{v}^{i,0}}$. Generate two NC weights according to NC2 in Definition \ref{nc_def}.\\
				
				\FOR{epoch index $k = 0, 1, \ldots, K$} 
				\STATE \texttt{--Local learning--} 
				\FOR{each client ${i \in [J]}$ in parallel}
				\STATE Updates its local model by $\mathbf{v}^{i,0} = \mathbf{v}^{k}$.
                \FOR{local update iteration $u=0,1,\ldots,U$}
				\STATE Update  $\mathbf{v}^{i,u+1}=\mathcal{G}(\mathbf{v}^{i,u};\mathcal{D}^i)$ locally.	
                \ENDFOR                
				\ENDFOR 				
                \STATE \texttt{--Global aggregation--} .
				\STATE The server collects $\mathbf{v}^{i,U}$ from all clients  ${i\in[J]}$, updates the global model by ${\mathbf{v}^{k+1} = ({1}/{[J]})\sum_{i=1}^J \mathbf{v}^{i,U}}$, and broadcasts it to all clients.
				\ENDFOR 		
			\end{algorithmic}
		\end{algorithm}

        The overall NCDSFL algorithm is given in \textbf{Algorithm} \ref{NCDSFL}, where $\mathcal{G}(.;.)$ in step 7 denotes a local update step based on some optimization criterion, such as gradient descent. One advantageous property of the NCDSFL algorithm is the exploration of two fixed NC weights, which helps guide and align the hidden feature distributions across clients. This promotes consistency among local models and mitigates the effects of data heterogeneity.
       Another key benefit lies in its reduced computational complexity and communication overhead. By fixing the auxiliary classifiers and output layer weights, the algorithm requires fewer variables to be updated and transmitted, making it more practical for real-world deployment.

    \subsection{Application into Signal Detection}\label{section_sys_model}

         The signal detection task can be formulated as a multi-binary classification problem by treating the prediction of each bit as a binary classification problem. The time-domain transmission model of the OFDM system can be formed by
        \begin{equation}\label{time_domain}
            b(t) = h({t,}\tau)*a(t)+u(t),
        \end{equation}
        where ${b(t)\in \mathbb{C}}$ and ${a(t)\in \mathbb{C}}$ denote the received signal and the time-domain transmit signal at time $t$, respectively, $u(t)\in \mathbb{C}$ is the additive Gaussian white noise (AWGN) at the receiver side,
        ${\{h({t},\tau):\tau=0,1,\ldots,Q-1\}}$ is the channel impulse response sequence with a length of $Q$ {at time $t$}, which is generated from the PDPs defined in the wireless world initiative for new radio (WINNER II) channel model \cite{kyosti2007winner}.  After removing the cyclic prefix signals,  $b(t)$ is passed through a discrete Fourier transform (DFT) module, and we have the following corresponding frequency-domain transmission expression of \eqref{time_domain},
        \begin{equation}\label{freq_domain}
            B(r) = A(r)H(r) + U(r),
        \end{equation}
        where $B(r)$, $A(r)$, $H(r)$ and $U(r)$ are the DFT of $b(t)$, $a(t)$, $h(t{,\tau})$ and $u(t)$, respectively.

        We adopt the commonly used multi-path channel model \cite{kyosti2007winner} to characterize the OFDM channels, namely,
        \begin{equation}\label{channel_model}
            h(t, \tau) = \sum_{m=0}^{M-1}\sqrt{P_m}A_m(t)\exp(-j\phi_m(t))\delta(\tau - \tau_m),
        \end{equation}
        where $P_m\in\mathbb{R}$ denotes the channel power gain over path $m$, $A_m\in \mathbb{C}$ characterizes the attenuation fluctuation and the phase shift of path $l$, which depends on the initial phase shift, the central frequency, the small scale fading, and the shadowing of path $m$, $\tau_m$ denotes the propagation delay of path $m$. 
        These parameters highly depend on the PDPs of users.
        Note that in this work, we focus on the low-mobility case. We therefore ignore the time variance caused by the user's mobility in \eqref{channel_model}.

        In wireless communication systems, the PDPs of different users depend on the relative positions among users, base stations, and scatterers. 
        Therefore, in cases of low mobility, as illustrated in Fig. \ref{sys_diagram}, each user can only observe the wireless channels characterized by one PDP within a given time period, resulting in the local observability problem.
        Users need to cooperate to learn a general model that works for newly joined users while preserving data privacy.
        We thus adopt the idea of FL to learn the general model, i.e.,  model \eqref{FL_opt},  for all users in a distributed manner.    
        
        Mathematically, each user (i.e., client) $i$ trains $E$ DNNs to predict bits at different positions of the transmit bit stream and these $E$ DNNs use the same training dataset with the same inputs and different labels $\mathcal{D}^{i,e}=\{(\mathbf{x}^i_1, \mathbf{y}^{i,e}_1),\cdots,(\mathbf{x}^i_N, \mathbf{y}^{i,e}_N)\}$, where $\mathbf{x}^i_n\in\mathbb{R}^{2d_x^i}$ and $\mathbf{y}^{i,e}_n\in\mathbb{R}^{2d_y^{i,e}}$  are calculated by received signal $b^i(t)$ and transmit signal $A^i(r)$ as follows, 
         \begin{align}\label{del-xin}
         \mathbf{x}^i_n =\Big(&{\rm Re}(b^i(t_n^i)),{\rm Im}(b^i(t_n^i)),\nonumber\\
         &{\rm Re}(b^i(t_n^i+1)),{\rm Im}(b^i(t_n^i+1)),\cdots,\\
         &{\rm Re}(b^i(t_n^i+d_x^i-1)),{\rm Im}(b^i(t_n^i+d_x^i-1))\Big)^T,\nonumber\\\label{del-yin}
           \mathbf{y}^{i,e}_n =\Big(&{\rm Re}(A^i(r_n^{i,e})),
           {\rm Im}(A^i(r_n^{i,e})),\nonumber\\
           &{\rm Re}(A^i(r_n^{i,e}+1 )),{\rm Im}(A^i(r_n^{i,e}+1 )),\cdots,\\
         &{\rm Re}(A^i(r_n^{i,e}+d_{y}^{i,e}-1)),{\rm Im}(A^i(r_n^{i,e}+d_{y}^{i,e}-1))\Big)^T, \nonumber
\end{align} 
where ${\rm Re}(b)$ and ${\rm Im}(b)$ represent the real and imagery part of $b\in\mathbb{C}$, and $r_n^{i,e+1}= r_n^{i,e}+d_{y}^{i,e}$. The above formulation means that we use  $d_x^i$ received signals $\{b^i(t): t\in t_n^i-1+[d_x^i]\}$  to generate the $n$th input $\mathbf{x}^i_n$ and  $d_y^{i,e}$ transmit signals $\{A^i(r): r\in r_n^{i,e}-1+[d_y^{i,e}]\}$ to generate the $n$th label $\mathbf{y}^{i,e}_n$ of the $e$th DNNs. Denote the total length of the label for user $i$ by $d_y^{i}=\sum_{e\in[E]}d_y^{i,e}$. After specifying local training data $\mathcal{D}^i$ for each user $i$, we apply NCDSFL to solve model \eqref{FL_opt}, yielding a solution $\mathbf{v}^i_*\approx\mathbf{v}_*,i\in[J]$, where $\mathbf{v}^i_*$ represents all the weights of the $E$ DNNs for client $i$. These trained networks can then be used to detect newly emerged signals based on the learned representations.

        \section{Simulation Results} \label{section_simulation}
        In this section, we conduct some experiments on joint channel estimation and symbol detection to evaluate the performance of the proposed NCDSFL algorithm. We adopt the B1 non-line-of-sight (NLoS) scenario from the  WINNER II model \cite{kyosti2007winner} as the channel generator.
Specifically, the channel consists of 24 propagation paths. The carrier frequency is set to 2.6 GHz with a bandwidth of 20 MHz.
We use typical urban channel conditions with a maximum delay of 16 sampling periods and a shadow fading standard deviation of 4 dB.
The delay spread follows a log-normal distribution with a mean of $-7.12 (\log_{10}[s])$ and a standard deviation of $ 0.12 (\log_{10}[s])$, where $s$ represents `second'.

        There are 10  users cooperating to train a general model so that newly joined users can also use this model without retraining.
        Client users only have access to their own data, and they can only observe a partial wireless environment.
        Therefore, we use the same PDP to generate 500 channels for the same user with randomly generated small scale channel parameters.
        These channel realizations are then fixed as a local dataset during the whole training process.
        The transmit data bit streams are randomly generated for each local training iteration, while the pilot bits are fixed in both training and testing phases.
        We test the bit error rate (BER) performance of the trained model on a general testing dataset generated from all PDPs with different receiving SNRs.

For each user ${i\in[10]}$, ${E=4}$ fully connected DNNs are trained, and each DNN consists of three hidden layers with 500, 250, and 128 neurons, using the ReLU function as activation function $\sigma$. The $n$th sample ${(\mathbf{x}^i_n,\mathbf{y}^i_n)}$ is generated as \eqref{del-xin} and \eqref{del-yin} with ${d_x^i=128}$ (including pilot signal) and  ${d_y^{i,e}=16}$. In Algorithm \ref{NCDSFL}, Each local training epoch consists of 50 local iterations to update ${\mathbf{v}^{i,u+1}=\mathcal{G}(\mathbf{v}^{i,u};\mathcal{D}^i)}$ using the RMSprop optimizer\footnote{RMSprop is an unpublished adaptive learning rate algorithm proposed by Geoff Hinton in Lecture 6e of his Coursera Class.}.
        The regularization coefficient in \eqref{OP31-reg} is set as $\mu=0.5$.
        
       We compare the proposed NCDSFL algorithm with two baselines: (1) the FedAvg-based signal detector \cite{mcmahan2017communication}, denoted as FL, and (2) the independent learning-based signal detector, denoted as IL. Both baselines use the same learning hyperparameters as NCDSFL. We also compare our method with the conventional minimum mean squared error (MMSE) signal detection algorithm, which uses MMSE-based channel estimates to recover the transmitted symbols.

        \subsection{Validation of Theorem \ref{NC_theorem}}
  
        We first validate the NC phenomenon in the signal detection scenario.
        In this experiment, we focus on the training process of a single client without the FL manner. 
        To validate Theorem \ref{NC_theorem}, we track two quantities during the training process, 
    \begin{equation}
        \begin{aligned}     \theta(\mathbf{W}_0,\mathbf{W}_1)&=\left\|\frac{I\mathbf{W}^T_0\mathbf{W}_0}{\|\mathbf{W}^T_0\mathbf{W}_0\|} - \mathbf{I}\right\|+\left\|\frac{I\mathbf{W}^T_1\mathbf{W}_1}{\|\mathbf{W}^T_1\mathbf{W}_1\|} - \mathbf{I}\right\|,\\
     \vartheta(\mathbf{W})&=\left\|\frac{{\rm vec}(\mathbf{H})}{\|{\rm vec}(\mathbf{H})\|}-\frac{\left[\mathbf{A}^T,-\mathbf{A}^T\right]{\rm vec}(\mathbf{W})}{\|\left[\mathbf{A}^T,-\mathbf{A}^T\right]{\rm vec}(\mathbf{W})\|}\right\|,
\end{aligned}    
    \end{equation}        
where $\mathbf{W}_0=[\mathbf{w}_{1,0},\ldots,\mathbf{w}_{I,0}]$ and $\mathbf{W}_1=[\mathbf{w}_{1,1},\ldots,\mathbf{w}_{I,1}]$.
         The results are shown in Figs. \ref{nc2} and \ref{nc3}.     
                   \begin{figure}[!th]
			\centering\vspace{-5mm}
			\includegraphics[width=0.45\textwidth]{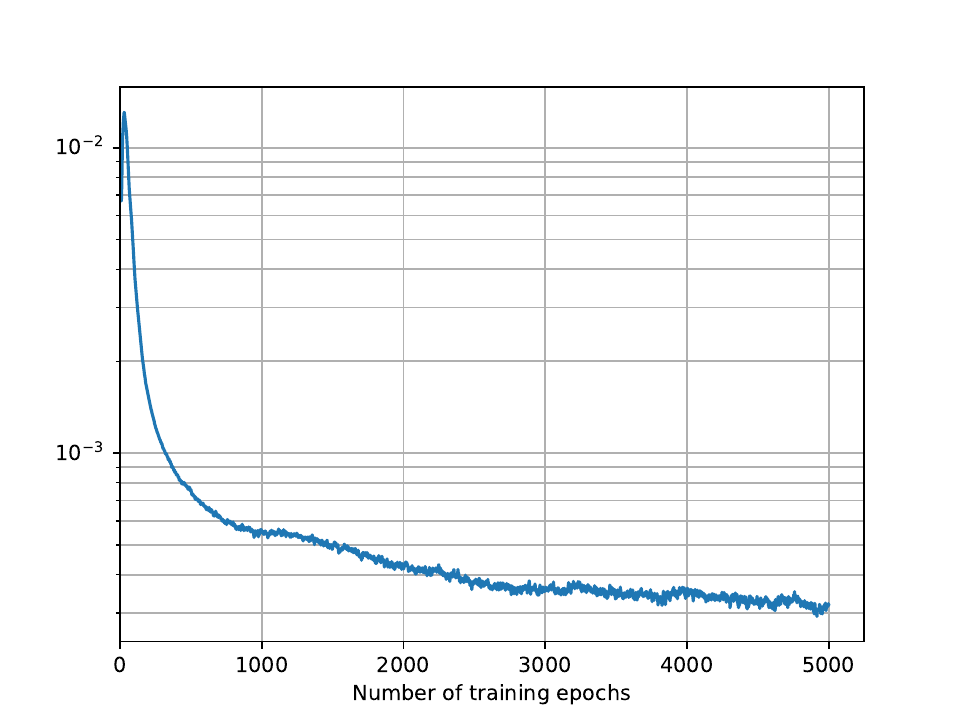}\vspace{-2mm}
			\caption{ $\theta(\mathbf{W}_0,\mathbf{W}_1)$ v.s.  epochs}
			\label{nc2} 
		\end{figure} 
		
		 The two figures show that both $ \theta$ and $\vartheta$ decrease monotonically throughout the training process. The decreasing trend of $ \theta$ suggests that the orthogonality among classifiers in the parallel binary classification tasks strengthens as training progresses. Similarly, the decline in $\vartheta$ indicates an increasing duality between the hidden features and their corresponding classifiers. The behavior of $ \theta$ supports our approach of employing fixed orthogonal classifiers, while the trend in $\vartheta$ implies that the hidden features gradually converge toward the subspace spanned by these orthogonal classifiers. The convergence underpins the feature alignment achieved via deep supervision using NC weights. Together, these confirm the presence of the NC phenomenon in the signal detection task.

            \begin{figure}[!th]
			\centering
			\includegraphics[width=0.45\textwidth]{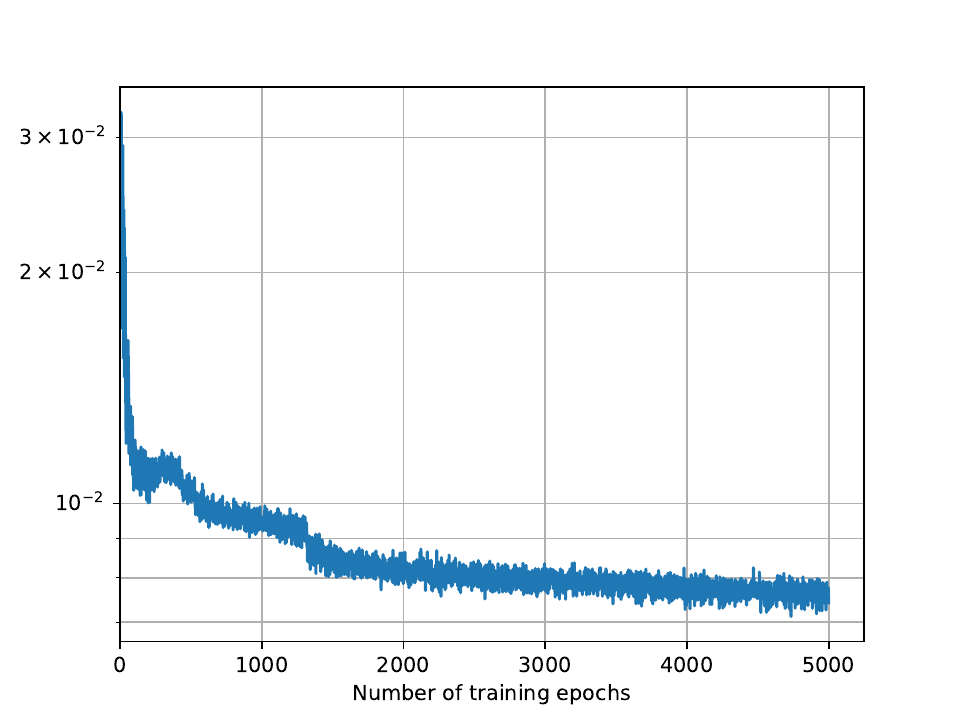}
			\caption{ $\vartheta(\mathbf{W})$ v.s. epochs}
			\label{nc3} 
		\end{figure}  
        
        \subsection{Performance Comparison with Baselines}
        Next, we compare the considered algorithms in our aforementioned FL learning scenario with a received signal SNR of 10dB.
        Fig. \ref{convergence} shows the BER performance on the validation dataset versus the number of training epochs.
        We can observe that the IL algorithm has the slowest convergence speed and worst performance, because each client only has partial observations of the wireless environment and cannot learn a general model for a general validation dataset.
        Our proposed NCDSFL algorithm and the FL algorithm have the same convergent BER performance.
        However, our proposed NCDSFL algorithm converges much faster than an FL algorithm.
        The NCDSFL algorithm converges at about 60 training epochs, while the FL algorithm converges at about 150 training epochs.
        This is because the hidden features of different clients are guided to a similar distribution defined by the NC weights as stated in Theorem \ref{NC_theorem}, which alleviates the impact of client drift effect and data discrepancy.
        Therefore, the NCDSFL algorithm can save communication rounds.
        
        \begin{figure}[!th]
			\centering
			\includegraphics[width=0.45\textwidth]{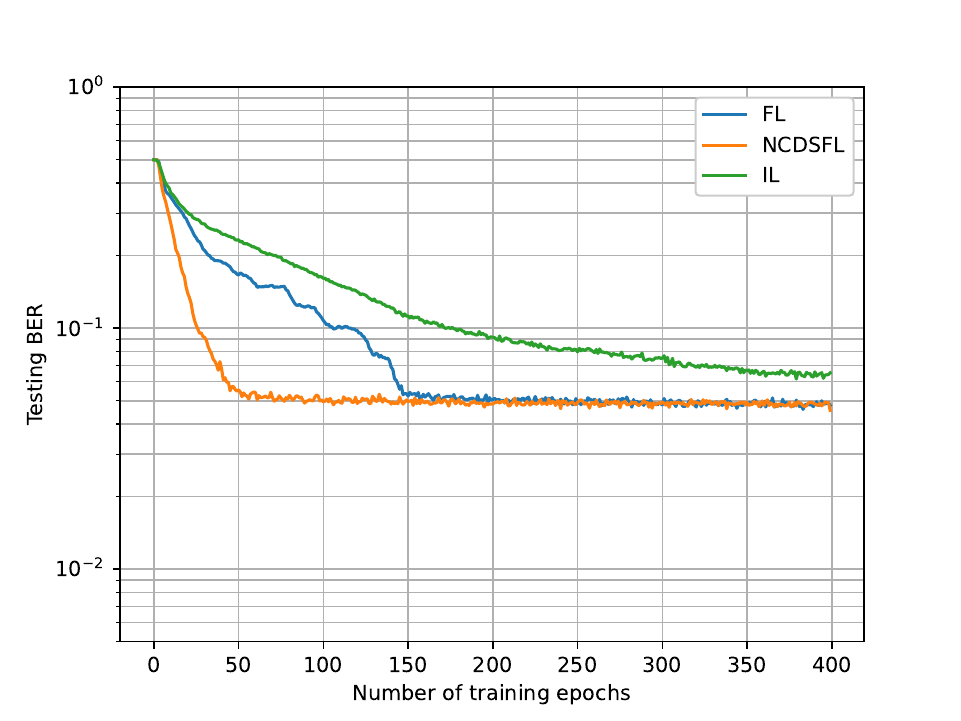}
			\caption{Testing BER v.s. epochs (SNR=10dB)}
			\label{convergence}
		\end{figure}

        In wireless systems, the large scale fading parameters, which highly depend on the relative positions of transmitters, receivers, and scatters, mainly determine the SNR of the received signal, even if the transmit power is adjustable.
        Therefore, we simulate a more practical scenario, where different client users receive signals with different SNRs.
        In this experiment, the SNRs are picked from $\{0,5,10,15,20\}dB$, with each SNR value corresponding to two client users.
        Fig. \ref{snr} draws the testing BER curves of the FL algorithm and the NCDSFL algorithm under different SNR conditions during the training process.
        We can observe from the figure that our proposed NCDSFL algorithm converges faster than an FL algorithm. 
        In addition, as the client heterogeneity increases, the convergent BER of our proposed NCDSFL algorithm also outperforms the FL algorithm.
        This validates the effectiveness of aligning clients by NC based DS in FL.
        \begin{figure}[!t]
			\centering
			\includegraphics[width=0.45\textwidth]{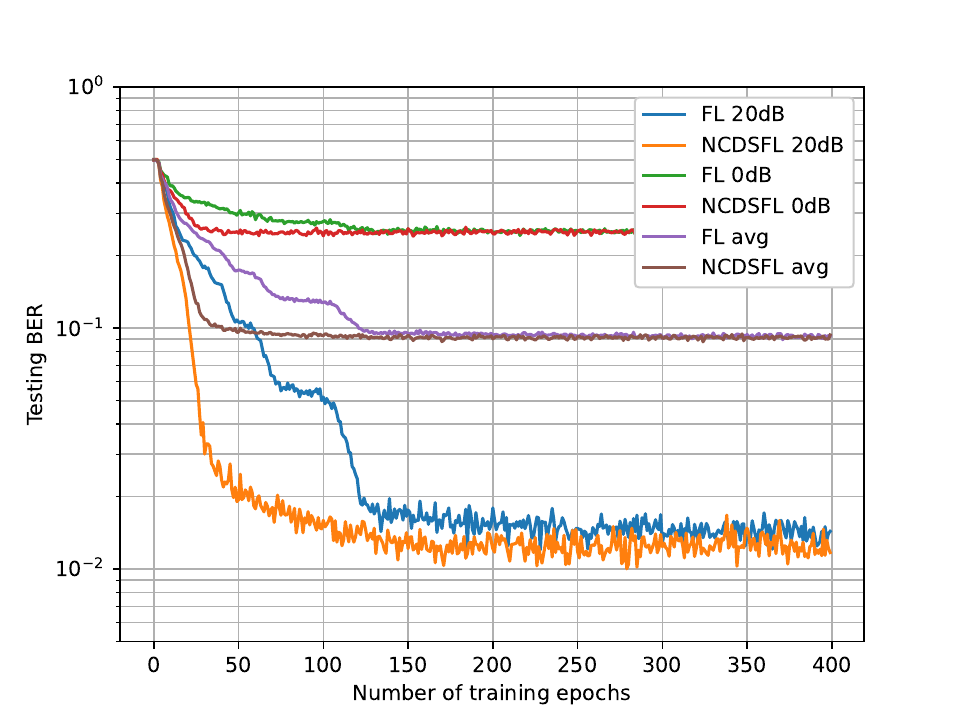}
			\caption{Testing BER v.s.  epochs }
			\label{snr}
		\end{figure}
        We summarize the testing performance of different algorithms versus SNR in Fig. \ref{snr_snr}.
        We can observe that our proposed NCDSFL algorithm always has equal or better testing performance than the FL algorithm.
        Moreover, the testing performance gap between the FL algorithm and the NCDSFL algorithm enlarges with the received signal SNR value.
        In addition, both the FL algorithm and the NCDSFL algorithm outperform the conventional MMSE baseline. This is because, due to a limited number of pilots, the MMSE algorithm cannot estimate the channels accurately and thus has a worse BER performance, while both the FL algorithm and the NCDSFL algorithm can learn the implicit OFDM channel structure from the received signals.
        \begin{figure}
			\centering
			\includegraphics[width=0.45\textwidth]{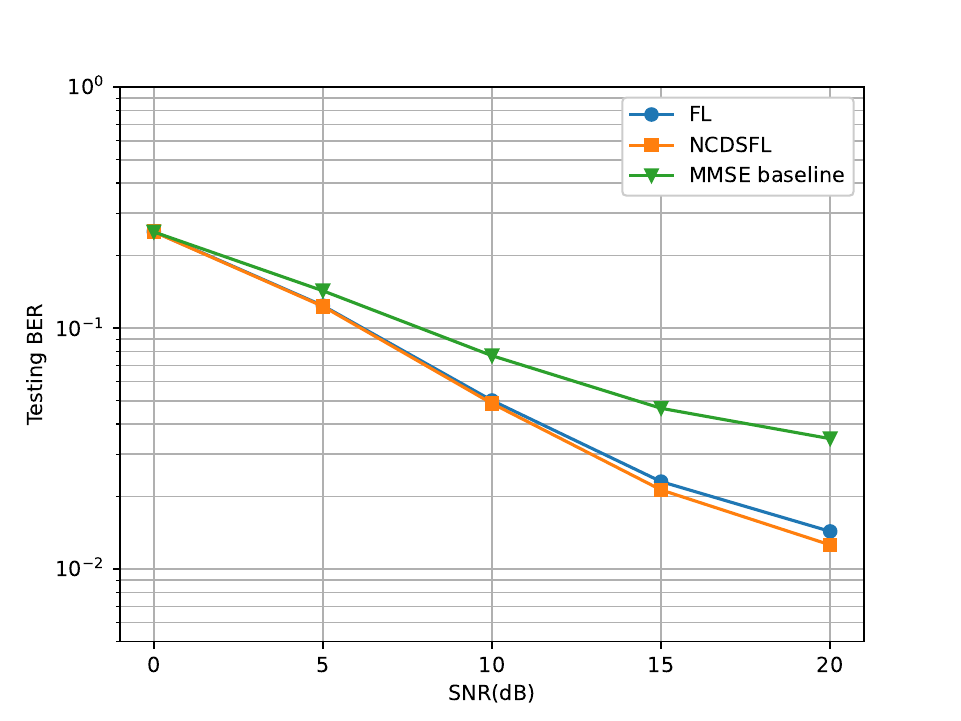}
			\caption{Testing BER v.s. SNR}
			\label{snr_snr}
		\end{figure}

        In the next experiment, we further investigate the impact of the LoS path.
        We add a LoS path to each client's channel and set different Rician factors for different users. The SNRs are set the same as in the last experiment, while the Rician factors are picked from $\{-5, 0, 5,10, 15\}dB$, with each value corresponding to two client users.
        Fig. \ref{rician_factor} draws the validation BER curve of different algorithms with different SNRs.
        We can observe that our proposed NCDSFL algorithm still converges faster than the FL algorithm.
        Besides, the BER performance of all algorithms is improved as the presence of LoS paths makes the signal detection problem easier to solve.
        \begin{figure}
			\centering
			\includegraphics[width=0.45\textwidth]{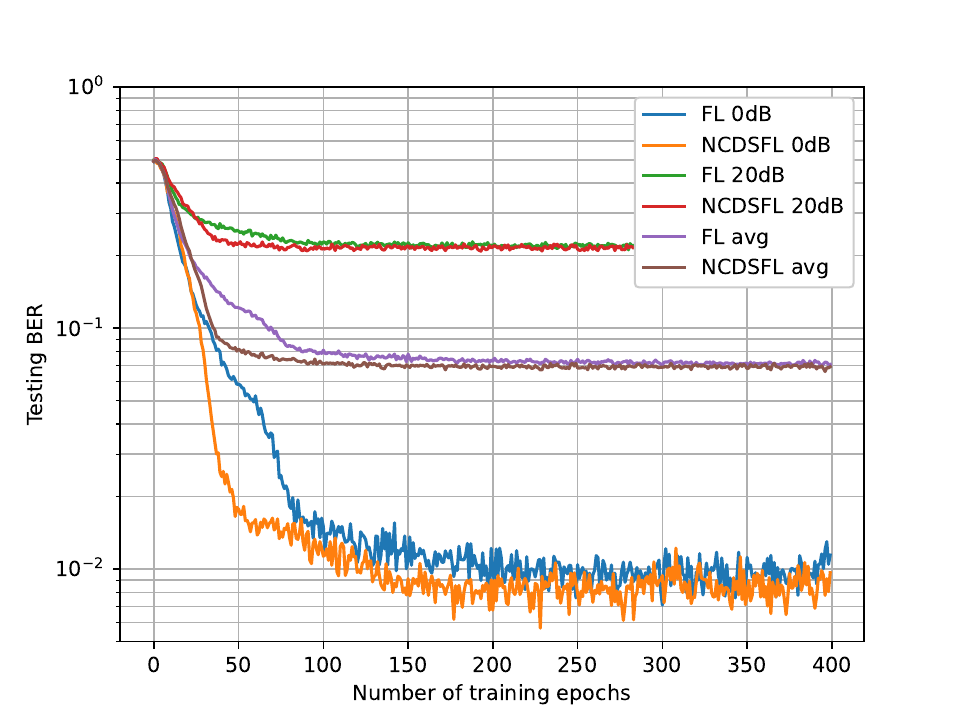}
			\caption{Testing BER v.s. epochs}
			\label{rician_factor}
		\end{figure}
        Fig. \ref{rician_factor_snr} draws the testing BER performance of different algorithms with different SNRs.
        We can observe that the testing performance of all algorithms is improved due to the presence of LoS path, and the learning based methods outperform the MMSE baseline due to their ability to learn the implicit information from the training data.
        \begin{figure}
			\centering
			\includegraphics[width=0.45\textwidth]{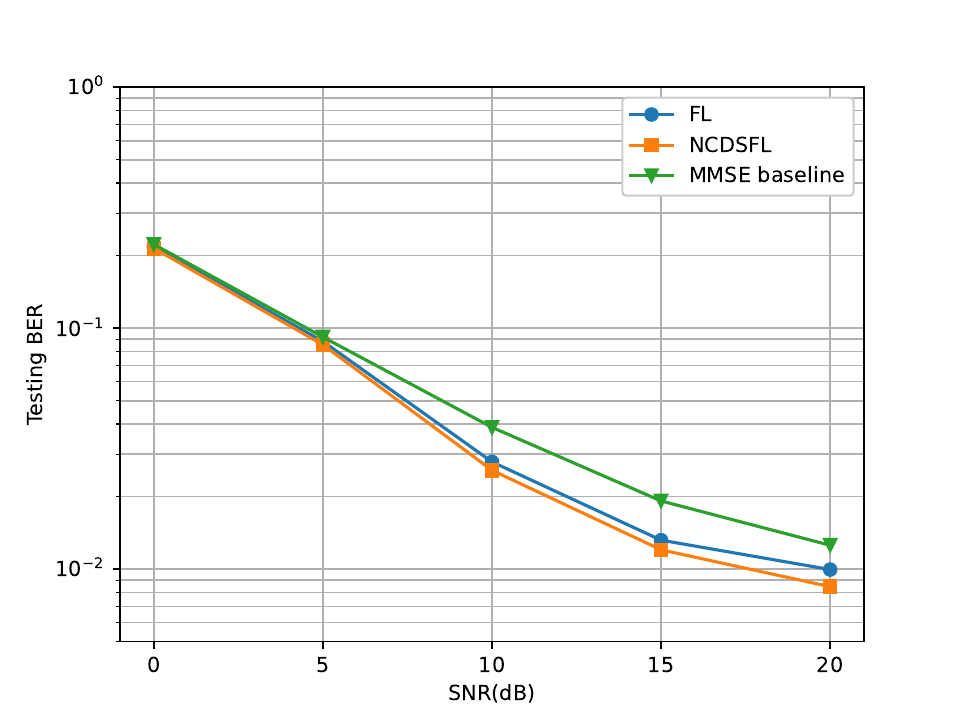}
			\caption{Testing BER v.s. SNR}
			\label{rician_factor_snr}
		\end{figure}

        In the final experiment, we introduce 5 additional clients whose channels are Rician channels with a Rician factor of 0dB.
        The SNR is set to 10dB.
        Fig. \ref{mixed_channels} draws the BER performance on WINNER II dataset and on Rician dataset of the trained model during the training process.
        We can observe that our proposed NCDSFL algorithm converges faster than the FL algorithm, due to the more efficient aggregation of the NCDSFL algorithm.
        Moreover, the performance of Rician channels is worse than the WINNER II channel.
        In Rician channels, except for the LoS path, the NLoS components are irrelevant.
        Therefore, with a limited number of pilots, it is more difficult to estimate the channels across all subcarriers for Rician channels compared with WINNER II channels.
        \begin{figure}
			\centering
			\includegraphics[width=0.45\textwidth]{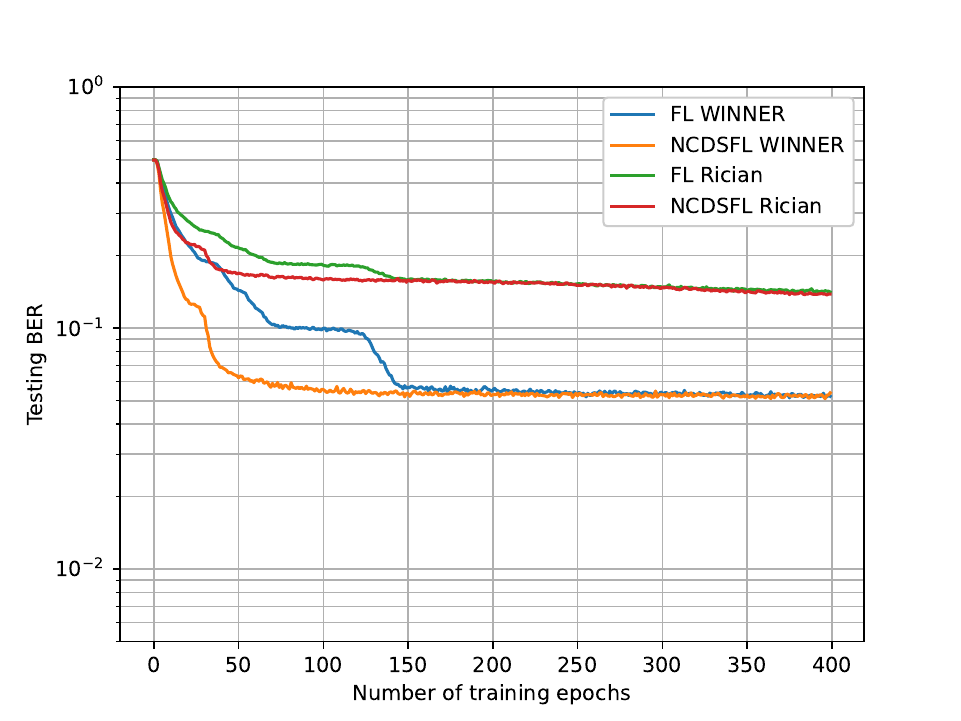}
			\caption{Testing BER  v.s. epochs (SNR=10dB)}
			\label{mixed_channels}
		\end{figure}
        
    \section{Conclusion} \label{section_conclusion}
        
      In this paper, we developed the NCDSFL algorithm by integrating the NC solutions with the DS scheme within an FL framework. This approach not only enables the effective aggregation of local knowledge from clients but also mitigates the impact of data heterogeneity, thereby accelerating convergence. Simulation results on joint channel estimation and symbol detection in OFDM systems demonstrate the efficiency of the proposed algorithm. Beyond this specific application, the NCDSFL algorithm holds potential for broader use in other domains, such as computer vision, which warrants further investigation in future research.

\begin{appendices}
\section{Proof of Theorem \ref{NC_theorem}} \label{proof}

        \subsection{Useful Lemmas}
        \begin{lemma}\label{row-orthogonal}
            Matrix $\mathbf{A}$ is row-orthogonal and has an equal norm of $\sqrt{K2^I}$ for each row.
        \end{lemma}
        \begin{proof}
            Since $\mathbf{C}^i (\mathbf{C}^i )^T = 2^{i} \mathbf{I}$, each matrix $\mathbf{C}^i,i=1,2,\ldots,I$ is row-orthogonal with equal norm row vectors. Now, we prove $\mathbf{B}^i$ for each ${i=1,2,\ldots,I}$ is row-orthogonal by mathematical induction.            
            When $i=1$, $\mathbf{B}^1 (\mathbf{B}^1)^T = 2^1 \mathbf{I}$.
            Assume that $\mathbf{B}^i$ satisfies $\mathbf{B}^i (\mathbf{B}^i)^T = 2^{i} \mathbf{I}$ for any ${i=1,2,\ldots I-1}$. For $i=I$, it follows from \eqref{B_def} that
                        \begin{equation}
            \begin{aligned}
           \mathbf{B}^{I}\left(\mathbf{B}^{I} \right)^T
           &=  \begin{bmatrix}
                2\mathbf{B}^{I-1}\left(\mathbf{B}^{I-1} \right)^T &\mathbf{0} \\
                 \mathbf{0} &2\mathbf{C}^{I-1}\left(\mathbf{C}^{I-1} \right)^T
            \end{bmatrix},\\
            &=  \begin{bmatrix}
               2^{I} \mathbf{I} &\mathbf{0} \\
                 \mathbf{0} & 2^{I} \mathbf{I} 
            \end{bmatrix}   = 2^{I} \mathbf{I}.
            \end{aligned}
            \end{equation}
 Using this fact and   \eqref{A_def}, we have 
            \begin{equation}\label{AAT-K2I}
            \begin{aligned}
          \mathbf{A}\mathbf{A}^T  &=  K \mathbf{B}^I\left(\mathbf{B}^I \right)^T =  K 2^{I} \mathbf{I},
            \end{aligned}
            \end{equation}
finishing the proof.
        \end{proof}
        For simplicity, let $\mathbf{I}_{dI}\in\mathbb{R}^{dI\times dI}$ and
        \begin{align} \label{w_def}
         \mathbf{h} &= {\rm vec}(\mathbf{H}),\nonumber\\
          \mathbf{w}&=[\mathbf{I}_{dI}, -\mathbf{I}_{dI}] {\rm vec}(\mathbf{W})\\
          &= \Big[(\mathbf{w}_{1,0}-\mathbf{w}_{1,1})^T,\ldots, (\mathbf{w}_{I,0}-\mathbf{w}_{I,1})^T\Big]^T.\nonumber
        \end{align}
Based on definitions of $\mathbf{H}$,  $\mathbf{W}$, and $\mathbf{w}$ in \eqref{def-W-H} and \eqref{w_def},  one can represent $P$ as a bilinear function of $\mathbf{w}$ and $\mathbf{h}$ as follows,
        \begin{equation}\label{P_def}
            P=\sum_{k=1}^K\sum_{i =1}^I\sum_{s\in\mathcal{S}}\phi_{kis}(\mathbf{W},\mathbf{H})=\mathbf{w}^T\mathbf{A}\mathbf{h},
        \end{equation}
        where $\mathbf{A}$ is defined in \eqref{A_def}.
\begin{lemma}\label{p_inequality}
    For $P$ defined in (\ref{P_def}), it holds
    \begin{equation}
        P\geq -\sqrt{K2^{I-1}} (\|\mathbf{W}\|^2+\|\mathbf{H}\|^2).
    \end{equation}
    The equality holds iff the following conditions are met:
    \begin{itemize}
   \item[1)]  $\mathbf{w}=-c\mathbf{Ah}$ for some constant $c>0$; 
    \item[2)]   $\mathbf{h}$ is a linear combination of row vectors of $\mathbf{A}$; 
     \item[3)]   $\mathbf{w}_{i,0}+\mathbf{w}_{i,1}=\mathbf{0},\forall i$; 
      \item[4)]   $\|\mathbf{W}\|=\|\mathbf{H}\|$.
    \end{itemize}
    
\end{lemma}
\begin{proof} Direct verification leads to
    \begin{equation}
    \begin{aligned}
        P=\mathbf{w}^T\mathbf{A}\mathbf{h}\overset{(a)}{\geq} &-\|\mathbf{w}\|\|\mathbf{A}\mathbf{h}\|\\
         \overset{(b)}{\geq} & -\sqrt{K2^I}\|\mathbf{w}\|\|\mathbf{h}\|\overset{(c)}{=}-\sqrt{K2^I}\|\mathbf{w}\|\|\mathbf{H}\|\\
               \overset{(d)}{\geq}& -\sqrt{2K2^I}\|\mathbf{W}\|\|\mathbf{H}\|\\
                \overset{(e)}{\geq} & -\sqrt{K2^{I-1}} (\|\mathbf{W}\|^2+\|\mathbf{H}\|^2),
    \end{aligned}       
    \end{equation}
    where (a) is because $\|\mathbf{a}\mathbf{b}\|\geq -\|\mathbf{a}\|\|\mathbf{b}\|$, and the equality in (a) holds iff $\mathbf{a}=-c\mathbf{b}$ for some constant $c>0$, namely, condition 1) holds;  (b) is due to $\|\mathbf{A}\mathbf{h}\|\leq\|\mathbf{A}\|_2\|\mathbf{h}\|$ and $\|\mathbf{A}\|_2 =\sqrt{K2^I}$ from Lemma \ref{row-orthogonal}, where $\|\mathbf{A}\|_2$ is the spectral norm of $\mathbf{A}$, and the equality in (b) holds iff condition 2) is met.   (c) is from the definition of $\mathbf{H}$ in \eqref{def-W-H};  (d) is from definitions of $\mathbf{W}$ and $\mathbf{w}$ in \eqref{def-W-H} and \eqref{w_def}, and the equality in (d) holds iff condition 3) is met. The equality in (e) holds iff condition 4) is met.
\end{proof}        

\begin{lemma}\label{nc_eq_lemma}
    If $\phi_{kis}(\mathbf{W},\mathbf{H})= c_1, \forall s \in \mathcal{S}, \forall i$ for a constant $c_1\leq0$, and the equality conditions of Lemma \ref{p_inequality} hold, then the NC conditions in Definition \ref{nc_def} hold.
\end{lemma}
\begin{proof}

        $\mathbf{h}$ is a linear combination of row vectors of matrix $\mathbf{A}$.
        Together with definition \eqref{w_def} of $\mathbf{h}$  and definition  \eqref{A_def} of $\mathbf{A}$, we can obtain NC1, the hidden features collapse to their corresponding feature means.

        Then we prove NC2. Based on condition 2) from Lemma \ref{p_inequality}, there is $\mathbf{z}$ such that $\mathbf{h}=\mathbf{A}^T\mathbf{z}$, which by condition 1) from Lemma \ref{p_inequality} results in 
        \begin{equation}
            \mathbf{w}=-c\mathbf{Ah} = -c\mathbf{A}\mathbf{A}^T\mathbf{z} = -c K2^I \mathbf{z}.
        \end{equation}
     where the last equation is from \eqref{AAT-K2I}. This suffices to
  \begin{align}\label{h-Aw}\mathbf{h}=\mathbf{A}^T\mathbf{z}=-\frac{1}{cK2^I}\mathbf{A}^T\mathbf{w}. \end{align}
By denoting $\triangle \mathbf{w}_{i}=\mathbf{w}_{i,1}-\mathbf{w}_{i,0}$, based on the definitions of $\mathbf{A}$ and $\mathbf{w}$ in \eqref{A_def} and \eqref{w_def}, the above condition yields
\begin{equation}
            \mathbf{h}_{s}^{(k)}=\frac{1}{cK2^I}\sum_{i=1}^I (2s_i-1)\triangle \mathbf{w}_{i}.
        \end{equation}
        Then the inner product condition becomes, for $\forall s, i$
        \begin{align}\label{wi}
          c_1= &\phi_{kis}(\mathbf{W},\mathbf{H}) \nonumber\\
           =&\left\langle-(2s_i-1)\triangle \mathbf{w}_{i},\mathbf{h}_s^{(k)}\right\rangle\nonumber \\
            =& -\frac{1}{cK2^I}\left\langle (2s_i-1)\triangle \mathbf{w}_{i}, \sum_{j}(2s_j-1)\triangle \mathbf{w}_{j}\right\rangle.
        \end{align}     
 Consider two binary sequences $s$ and $t$, which are same except for the $k$-th position, i.e., $s_i=t_i, \forall i\neq k$ and $s_k+t_k=1$.
        We have the following equalities for $\forall i$, 
        \begin{equation}\label{w1}
             -\frac{1}{cK2^I}\left\langle \left(2s_i-1\right)\triangle \mathbf{w}_{i}, \sum_{j=1}^I \left(2s_j-1\right)\triangle \mathbf{w}_{j}\right\rangle=c_1, 
        \end{equation}
        \begin{equation}\label{w2}
             -\frac{1}{cK2^I}\left\langle\left( 2t_i-1\right)\triangle \mathbf{w}_{i}, \sum_{j=1}^I\left(2t_j-1\right)\triangle \mathbf{w}_{j}\right\rangle=c_1.
        \end{equation}
        For $\forall i \neq k$, subtracting the above two equations leads to 
        \begin{equation} 
             \left(2t_i-1\right)\left\langle \triangle \mathbf{w}_{i}, 2\left(s_k-t_k\right)\triangle \mathbf{w}_{k}\right\rangle=0.
        \end{equation}
  indicating that
        \begin{equation}
            \left\langle \triangle \mathbf{w}_i, \triangle \mathbf{w}_k\right\rangle=0, \forall i \neq k.
        \end{equation}
        Together with condition 3) of Lemma \ref{p_inequality}, we can conclude that $\{\mathbf{w}_{i,0}\}$ and $\{\mathbf{w}_{i,1}\}$ are two sets of orthogonal vectors.
        Furthermore, by applying the orthogonality to \eqref{wi}, we have 
        \begin{equation}
            -\frac{1}{cK2^I} (2s_i-1)^2\left\langle \triangle \mathbf{w}_{i},\triangle \mathbf{w}_{i}\right\rangle=c_1,
        \end{equation}
        which implies that 
        \begin{equation}
            \|\triangle \mathbf{w}_i\|^2=-cc_1K2^I.
        \end{equation}
        This completes the proof of NC2. Finally,  NC3 is obtained by \eqref{h-Aw}.
        This completes the proof.
\end{proof}

\subsection{Proof of Theorem \ref{NC_theorem}}
    As function $f(x)=\ln(1+\exp(x))$ is  convex, we can apply Jensen's inequality to the following function,
        \begin{equation}
            \begin{aligned}
            l(\mathbf{W},\mathbf{H}) &=\frac{1}{KI2^I}\sum_{k=1}^K\sum_{i=1}^I \sum_{s\in\mathcal{S}}\ln\Big(1+\exp\Big(\phi_{kis}(\mathbf{W},\mathbf{H})\Big)\Big)\\
            &\geq \ln\left(1+\exp\left(\frac{1}{KI2^I}\sum_{k=1}^K\sum_{i=1}^I \sum_{s\in\mathcal{S}}\phi_{kis}(\mathbf{W},\mathbf{H})\right)\right)\\
            &= \ln\left(1+\exp\left(\frac{P}{KI2^I}\right)\right),
            \end{aligned}
        \end{equation}
       where the equality in `$\geq$' holds iff $\phi_{kis}(\mathbf{W},\mathbf{H})= c_1, \forall s \in \mathcal{S}, \forall i$. Let $\rho = \|\mathbf{W}\|^2+\|\mathbf{H}\|^2$. From Lemma \ref{p_inequality}, it holds
        \begin{equation}\label{P_LB}
            \begin{aligned}
                P \geq -\sqrt{K2^{I-1}} \rho.
            \end{aligned}            
        \end{equation}
Using the above facts, we have
        \begin{equation}
        \begin{aligned}
            L(\mathbf{W},\mathbf{H}) &=l(\mathbf{W},\mathbf{H}) + \lambda \rho\\
           &  \geq  \ln\left(1+\exp\left(\frac{P}{KI2^I}\right)\right) +\lambda \rho\\
           & \geq \ln\left(1+\exp\left(-t\rho\right)\right)+\lambda \rho = L(\rho),
        \end{aligned}            
        \end{equation}
where ${t=1/(I\sqrt{2K2^I}})$. It is clear that $L(\rho)$ is a convex function of $\rho$ and problem $\min_\rho  L(\rho)$ has zero solution if ${\lambda \geq t/2}$. In other words, problem $\min_\rho  L(\rho)$ admits
a non-trivial global minimizer if ${0<\lambda <t/2}$. Using the first-order optimality condition, $L'(\rho_{opt})=0$, one can find optimal solution $\rho_{opt}$ by
        \begin{equation}
            \rho_{opt}=\frac{1}{t}\ln\frac{t-\lambda}{\lambda}.
        \end{equation}
        Then we have the following inequality,
        \begin{equation}
            L(\mathbf{H},\mathbf{W})\geq L(\rho_{opt}),
        \end{equation}
        which indicates that problem \eqref{multi-binary_classification} also admits non-trivial global minimizers. Recall the sufficient and necessary conditions for the global minimizers, i.e., $\phi_{kis}(\mathbf{W},\mathbf{H})= c_1, \forall s \in \mathcal{S}, \forall i$ and conditions in Lemma \ref{p_inequality}. Based on Lemma \ref{nc_eq_lemma}, we obtain Theorem \ref{NC_theorem}. \qed
           
\end{appendices}

		\bibliographystyle{IEEEtran}
		\bibliography{bib}

	\end{document}